\documentclass[conference]{IEEEtran}
\IEEEoverridecommandlockouts

\usepackage[english]{babel}
\usepackage{cite}
\usepackage{amsmath,amssymb,amsfonts}
\usepackage{algorithmic}
\usepackage{amsthm}
\usepackage{graphicx}
\usepackage{textcomp}
\usepackage{xcolor}
\usepackage{tikz}
\usetikzlibrary{patterns}
\usepackage{enumitem}
\usepackage{pgfplots}
\pgfplotsset{compat=1.18}
\usetikzlibrary{pgfplots.groupplots}
\usetikzlibrary{shapes.geometric, arrows.meta, positioning}
\def\BibTeX{{\rm B\kern-.05em{\sc i\kern-.025em b}\kern-.08em
    T\kern-.1667em\lower.7ex\hbox{E}\kern-.125emX}}

\definecolor{DarkGreen}{RGB}{0,150,0}

\newtheorem{theorem}{Theorem}
\newtheorem{lemma}[theorem]{Lemma}

\begin{document}

\newcommand{\bx}{\mathbf{x}}
\newcommand{\bX}{\mathbf{X}}
\newcommand{\by}{\mathbf{y}}
\newcommand{\bY}{\mathbf{Y}}
\newcommand{\bs}{\mathbf{s}}
\newcommand{\ba}{\mathbf{a}}
\newcommand{\br}{\mathbf{r}}
\newcommand{\bg}{\mathbf{g}}
\newcommand{\bn}{\mathbf{n}}
\newcommand{\bN}{\mathbf{N}}
\newcommand{\bw}{\mathbf{w}}
\newcommand{\bz}{\mathbf{z}}
\newcommand{\bh}{\mathbf{h}}
\newcommand{\boldg}{\mathbf{g}}
\newcommand{\bH}{\mathbf{H}}
\newcommand{\bG}{\mathbf{G}}
\newcommand{\bA}{\mathbf{A}}
\newcommand{\bB}{\mathbf{B}}
\newcommand{\bR}{\mathbf{R}}
\newcommand{\bC}{\mathbf{C}}
\newcommand{\bI}{\mathbf{I}}
\newcommand{\bQ}{\mathbf{Q}}
\newcommand{\bGamma}{\mathbf{\Gamma}}
\newcommand{\bgamma}{\gamma}

\title{Resilient Vital Sign Monitoring \\
Using RIS-Assisted Radar
\thanks{This work has been performed in the context of the LOEWE center emergenCITY [LOEWE/1/12/519/03/05.001(0016)/72].\\
We thank Ehsan Mohammadi and Hoang Oanh Pham for assisting with the vital sign recordings conducted during their master's and bachelor's thesis, respectively.}
}

\author{Christian Eckrich\textsuperscript{*,**}, Abdelhak M. Zoubir\textsuperscript{*}, and Vahid Jamali\textsuperscript{**}\\
\textsuperscript{*}Signal Processing Group, \textsuperscript{**}Resilient Communication Systems\\ Technische Universität Darmstadt, Germany\\
}

\maketitle

\begin{abstract}
Vital sign monitoring plays a critical role in healthcare and well-being, as parameters such as respiration and heart rate offer valuable insights into an individual's physiological state. While wearable devices allow for continuous measurement, their use in settings like in-home elderly care is often hindered by discomfort or user noncompliance. As a result, contactless solutions based on radar sensing have garnered increasing attention. This is due to their unobtrusive design and preservation of privacy advantages compared to camera-based systems. However, a single radar perspective can fail to capture breathing-induced chest movements reliably, particularly when the subject's orientation is unfavorable. To address this limitation, we integrate a reconfigurable intelligent surface (RIS) that provides an additional sensing path, thereby enhancing the robustness of respiratory monitoring. We present a novel model for multi-path vital sign sensing that leverages both the direct radar path and an RIS-reflected path. We further discuss the potential benefits and improved performance our approach offers in continuous, privacy-preserving vital sign monitoring.
\end{abstract}

\section{Introduction}
Continuous monitoring of vital signs like respiration and heart rate is crucial for early detection of life-threatening events. While intensive care units (ICUs) use advanced systems, such equipment is often lacking in general hospital units due to cost and complexity. This can delay the response to critical incidents. 
There is a clear need for non-contact, affordable, and reliable monitoring solutions that work continuously without causing patient discomfort. 

Radar-based sensing offers a promising solution for non-invasive vital sign monitoring by detecting subtle movements caused by respiration and heartbeat \cite{paterniani_radar-based_2022, mercuri_enabling_2021, eder_sparsity_2022, wang_mmhrv_2021}. However, results can become unreliable due to two main challenges: (1) physical obstructions such as medical equipment or bed positioning that block the direct path, and (2) poor alignment with the patient's chest caused by posture, leading to weak measurements. One way to address this issue is to deploy multiple radars, i.e., distributed radar networks (DRNs) \cite{ren_vital_2021}. Besides the additional cost of radar systems, DRNs require a fronthaul link to share their data for joint processing \cite{eckrich2024fronthaul}. We propose the use of reconfigurable intelligent surfaces (RISs) to overcome these issues. RIS can reflect and steer radar signals to create multiple paths, ensuring higher signal quality and accurate measurements \cite{tripathy_liquid_2025, mercuri_reconfigurable_2023}. This effectively enables distributed sensing with only one active radar, with RIS pathways providing spatially diverse perspectives. To the best of the authors' knowledge vital sign monitoring using RIS-assisted radar has not been explored yet.  
\begin{figure}
    \centering
    \resizebox{0.5\linewidth}{!}{
        \input{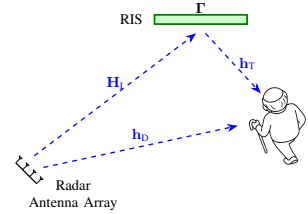}
    }
    \caption{Radar-based vital sign estimation assisted by an RIS.}
    \label{fig:SetupSketch}
\end{figure}
The main contributions of this paper are listed below as: 
\begin{itemize}[wide = 0pt]
    \item The dependence of the received vital sign-modulated radar signal on the angle of reflection from the human target is a complex phenomenon that cannot be fully captured by analytical models alone. To address this, we develop a hybrid experimental-analytical approach. Specifically, we model the well-known RIS-assisted multipath channel analytically, where we characterize the angle-dependent radar cross section (RCS) of the target (containing the vital sign information) through real-world measurements obtained from our experimental radar platform.
    \item Furthermore, to enable simultaneous illumination of the direct and RIS-assisted paths, we derive a closed-form minimum-norm transmit beamformer that fulfils two independent steering constraints while respecting the radar's power budget.
    We propose three algorithms for extracting vital signs from both the direct path and the RIS-generated path: (i) spatial-separation: in this algorithm, the radar transmits simultaneously towards both paths using appropriate beamformers. (ii) temporal separation: here, the radar alternates between the two paths in consecutive time slots to exploit both paths periodically. (iii) opportunistic access: this algorithm dynamically selects the best path at each moment depending on factors such as the target's posture and orientation. We introduce the necessary signal processing techniques for vital sign extraction, including receive beamforming, clutter removal, and phase demodulation.
    \item Finally, we present an extensive set of results evaluating the sensing performance of the proposed algorithms. These studies confirm that the path diversity enabled by RIS significantly enhances the accuracy of vital sign estimation. Furthermore, the results suggest that additional performance improvements are achievable through optimized resource allocation, including power, time, and path selection.
\end{itemize}

\section{System Model}

As depicted in Fig. \ref{fig:SetupSketch}, we consider a radar system consisting of $M$ trans-receivers that emit a sensing signal $\bX \in \mathbb{C}^{M\times K}$ of $K$ temporal samples towards a target and collect the backscattered signal $\bY \in \mathbb{C}^{M\times K}$.
We assume that self-interference is negligible, which can be achieved through full duplex radars \cite{zhang_self-interference_2022} or by using different but collocated transmit/receive antennas \cite{ahmed_all-digital_2015}. To provide path diversity, an RIS with $N$ unit-cells is deployed. 
The signal model is then given by
\begin{align}
    \bY = \bH(\alpha, \beta) \, \bX + \bN, 
\end{align}
where $\bY \in \mathbb{C}^{M\times K}$ is the received signal matrix and $\bN \in \mathbb{C}^{M\times K}$ is the sensor noise modelled as additive white Gaussian noise. Moreover, $\bH \in \mathbb{C}^{M\times M}$ denotes the end-to-end channel matrix between radar to a point target and back (including the impact of the RIS). Parameters $\alpha$ and $\beta$ explicitly indicate the target RCS observed through the direct path and the RIS-generated path, respectively. In particular, $\bH(\alpha, \beta)$ can be written as 
\begin{align}
    \bH(\alpha, \beta) = &\bH_\text{I}^\top\bGamma^\top\bh_\text{T}^\top \alpha \bh_\text{T} \bGamma \bH_\text{I}
    + \bh_\text{D}^\top \beta \bh_\text{D} + \bH_\text{C},
\end{align}
where $\bH_\text{I}\in \mathbb{C}^{M\times N}$ is the channel matrix between the radar and the RIS, $\bh_\text{T}\in \mathbb{C}^{N}$ is the channel vector between the RIS and the point target, and $\bh_\text{D}\in \mathbb{C}^{M}$ is the channel vector for the direct path. Moreover, $\bGamma\in \mathbb{C}^{N\times N}$ is a diagonal matrix representing the RIS reflection coefficients, i.e., $\bGamma = \text{diag}([e^{j\phi_1}, \hdots, e^{j\phi_N}])$, where $\phi_n$ is the phase shift applied by the $n$-th element of the RIS. Furthermore, $\bH_\text{C}\in \mathbb{C}^{M\times M}$ denotes the impact of objects cluttering the environment. The superscript $(\cdot)^T$ denotes transpose.
The channel components $\bH_\text{I}$, $\bh_\text{T}$, and $\bh_\text{D}$ are modeled as Rician fading channels with a Rician factor $K$, i.e.,
\begin{align}
    \sqrt{\frac{K}{K+1}}\bH_\text{LoS} + \sqrt{\frac{1}{K+1}}\bH_\text{nLoS}, 
\end{align}
where $\bH_\text{nLoS}$ is the nLoS component and its entries are distributed as $\mathcal{CN}(0,\bI)$, while $\bH_\text{LoS}$ is the LoS component of $\bH_\text{I}$, $\bh_\text{T}$, and $\bh_\text{D}$.
\subsection{RCS Model}
The complex RCSs of a person observed from the direct path and the RIS-assisted path are denoted by $\alpha$ and $\beta$, respectively. Periodic expansions and contractions of the thoracic cage during respiration modulate the phase of the back-scattered radar signal, leading to a variation along the slow-time scale. Note that, the intensity of the chest displacement during breathing is not uniform across all observation angles.
The exact influence of the incident angle is complex, as will be seen in the experimental section. It depends on the patient's shape, posture, breathing behavior, and the sensing signal itself. However, in general, we can observe that the breathing can be best observed when the radar is aligned with the chest of the person. With an increasing observation angle, the breathing signal becomes weaker and more difficult to detect. Observing the patient from the side leads to a complete loss of the breathing signal. Fig. \ref{fig:Respirtation_time_plot} plots the extracted chest displacement signals over time from our experimental platform. In the top plot, the front facing radar trace (dashed) is overlaid with that from a second radar at a similar incidence angle (solid). The two curves almost coincide, demonstrating virtually identical signal quality. In the bottom plot, the front facing trace (dashed) is compared with a radar at a much larger observation angle (solid). Here, the secondary signal is noticeably attenuated and distorted, and its characteristic respiratory oscillations are no longer visible.

For simplicity, we assume that the sensing signal is reflected from the person's chest at one scattering point. 
To explicitly model the slow variations of the RCS due to respiration, we introduce the subscript $l$, where $\alpha_l$ and $\beta_l$ denote the target RCS at the $l$-th sample taken on a slow-time scale.
The complex RCS of the target can be modeled as
\begin{align}
    \alpha_l = q_\alpha e^{j\frac{2 \pi}{\lambda} d_\alpha[l]}, \quad \beta_l = q_\beta e^{j\frac{2 \pi}{\lambda} d_\beta[l]},
\end{align}
where $q_\alpha$ and $q_\beta$ are the reflectivity coefficients of the target for the RIS and direct path, respectively. Moreover, $d_\alpha[l]$ and $d_\beta[l]$ are the distance variations due to respiration seen from both perspectives, as illustrated in Fig. \ref{fig:Respirtation_time_plot}.
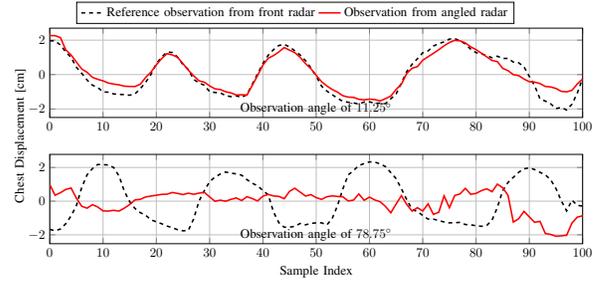
\begin{figure}
    \centering
    \resizebox{0.9\linewidth}{!}{
        \begin{tikzpicture}
\begin{groupplot}[
    group style={
        group size=1 by 2,       
        horizontal sep=0.5cm,        
    },
    width=16cm,
    height=4cm,
    xmin=0, xmax=100,
    grid=both,
    legend style={at={(0.05,1.05)},anchor=south west,legend columns=-1},
]

\nextgroupplot[
    title={Observation angle of $11.25^\circ$}, 
    title style={at={(0.5,0.01)},anchor=center},
    ylabel={Chest Displacement [cm]},
    y label style={at={(axis description cs:0,-0.20)},anchor=south,yshift=5mm}
]
\addplot[black, dashed, very thick]
    table [
        x expr=\thisrow{index}-70,
        y=front_radar_VS,
        col sep=comma,
        restrict x to domain=0:100
    ] {figures/data/respiration_time_1.csv};
\addlegendentry{\shortstack{Reference observation from front radar}}

\addplot[red, very thick]
    table [
        x expr=\thisrow{index}-70,
        y=side_radar_VS,
        col sep=comma, 
        restrict x to domain=0:100
    ] {figures/data/respiration_time_1.csv};
\addlegendentry{Observation from angled radar}

\nextgroupplot[
    title={Observation angle of $78.75^\circ$},
    xlabel={Sample Index},
    title style={at={(0.5,0.01)},anchor=center},
]
\addplot[black, dashed, very thick]
    table [
        x expr=\thisrow{index}-1,
        y=front_radar_VS,
        col sep=comma, 
        restrict x to domain=0:100
    ] {figures/data/respiration_time_7.csv};

\addplot[red, very thick]
    table [
        x expr=\thisrow{index}-1,
        y=side_radar_VS,
        col sep=comma, 
        restrict x to domain=0:100
    ] {figures/data/respiration_time_7.csv};
\end{groupplot}
\end{tikzpicture}
    }
    \caption{Chest displacement observed from small and large incident angle with respect to the chest's normal.}
    \label{fig:Respirtation_time_plot}
\end{figure} 
\section{Proposed Method}
In this section, we first detail the radar transmission design, covering the sensing waveform, dual-path beamforming, and power/time-allocation strategies. Subsequently, we describe the complete signal processing chain used to extract vital signs from the received echoes.

\subsection{Sensing Signal}
\label{SensingSignalAndBeamformer}
The adopted sensing signal is a pulsed sinusoidal wave of duration $T_p$ at frequency $f_0$ with bandwidth $1/T_p$. For practical reasons, it is necessary for the sensing signal to be narrowband since the RIS response time is much slower than the sweep time.
The signal is transmitted by the radar's transmit antennas at a pulse repetition interval of $T_\text{PRI}$. Hence, the sensing signal $s(t)$ can be modeled as
\begin{align}
    s(t) = \sqrt{2}\cos(2\pi f_0 t).
\end{align}
The sensing signal is sampled at a rate of $f_s$ so that the fast-time sweep $T_p$ consists of $K$ samples, i.e., $T_p = K/f_s$. The sampled sensing signal is then arranged in the vector $\bs \in \mathbb{C}^{K}$. Using a precoder $\bw \in \mathbb{C}^{M}$, the emitted sensing signal at the $M$ antenna elements of the radar array is given by 
\begin{align}
    \bX = \bw \bs^{H},
\end{align}
where the superscript $(\cdot)^H$ denotes Hermitian operation.
\subsection{Beamforming Design for Dual-Path Radar Transmission}
The direct path usually dominates and masks the signal from the RIS path because it experiences less path loss. To still exploit the spatial diversity of both paths, the transmit precoder $\bw$ must simultaneously control the array response in two distinct directions: towards the direct path and towards the RIS. This control allows either (i) exclusive focusing on one path (temporal separation), (ii) power allocation to both paths (spatial separation), or (iii) adaptive switching (opportunistic selection). 
The optimal precoder that minimizes the transmit power while maintaining the array response $|\ba_i^H\bw|$ towards the steering direction $\theta_i$ at $\gamma_i \in \mathbb{R}$ can be found by solving
\begin{subequations}
\begin{align}
    \textbf{P1:}& \quad \min_{\bw} \; \bw^H\bw\\
    \text{s.t.}& \quad |\ba_1^H\bw| = \gamma_1, \quad |\ba_2^H\bw| = \gamma_2
\end{align}
\end{subequations}
where $\ba_1$ and $\ba_2$ are the radar array responses. For a uniform linear array $\ba_i$, for $i = 1,2$, is given by 
\begin{align}
    \ba_i = \frac{1}{\sqrt{M}}
    \begin{bmatrix}
        1\\
        \exp{\left(j2\pi \frac{1}{\lambda_0}\delta \sin\theta_i\right)}\\
        \vdots\\
        \exp{\left(j2\pi \frac{1}{\lambda_0} (M-1) \delta \sin\theta_i\right)}
    \end{bmatrix}, 
\end{align}
where $\lambda_0$ is the wavelength of the radar signal and $\delta$ is the inter-element spacing of the antenna array. This problem can be reformulated as 
\begin{subequations}
\begin{align}
    \textbf{P2:}& \quad \min_{\bw, \phi_1, \phi_2} \; \bw^H\bw\\
    \text{s.t.}& \quad \ba_1^H\bw = \gamma_1 e^{j\phi_1}, \quad \ba_2^H\bw = \gamma_2 e^{j\phi_2},
\end{align}
\end{subequations}
where $\phi_1$ and $\phi_2$ denote arbitrary phase terms that account for the unknown phase of the array responses.
\begin{lemma}
    For a given $\phi_1$ and $\phi_2$, the optimal beamformer that satisfies the constraints above while minimizing the transmit power is
    \begin{align}
        \bw = \bA (\bA^H \bA)^{-1} \bg
    \end{align}
    where $\bA = [\ba_1, \ba_2]\in \mathbb{C}^{M\times 2}$ and $\bg = [\gamma_1 e^{j \phi_1}, \gamma_2 e^{j \phi_2}]^T$. \label{lem1}
\end{lemma}

\begin{proof}
    We minimize $\lVert\bw\rVert^2$ subject to the two linear constraints
    $\bA^{H}\bw=\bg$.  
    Introduce the Lagrange vector $\boldsymbol\lambda\in\mathbb{C}^{2}$ and form
    \begin{align}
        \mathcal{L}(\bw,\boldsymbol\lambda)=
    \bw^{H}\bw+\boldsymbol\lambda^{H}\!\bigl(\bg-\bA^{H}\bw\bigr)
        +\bigl(\bg^{H}-\bw^{H}\bA\bigr)\boldsymbol\lambda .
    \end{align} 
    Setting $\partial\mathcal{L}/\partial\bw^{*}=0$ gives
    $\bw=\bA\boldsymbol\lambda$.
    Substituting $\bw$ into $\bA^{H}\bw=\bg$ yields 
    \begin{align}
        \bA^{H}\bA\,\boldsymbol\lambda=\bg .
    \end{align}
    Because the two steering vectors are linearly independent,
    $\bA^{H}\bA$ is invertible, so
    $\boldsymbol\lambda=(\bA^{H}\bA)^{-1}\bg$.
    Inserting this back into $\bw=\bA\boldsymbol\lambda$ gives
    \begin{align}
        \bw=\bA(\bA^{H}\bA)^{-1}\bg
    \end{align}
    This expression is the Moore Penrose pseudoinverse solution, which is the unique solution for the constraint with minimum Euclidean norm. Hence, $\mathbf{w}$ is optimal. This completes the proof.
\end{proof}
\begin{lemma} \label{lem2}
    The optimal phases $\phi_1$ and $\phi_2$ are given by $\Delta \phi = \phi_1-\phi_2 = -\angle a_c$, where $a_c = \ba_1^H \ba_2$ is the steering vector correlation. Moreover, the minimum power is given by
    \begin{align}
        P_\text{min} = \frac{\gamma_1^2 + \gamma_2^2 - 2\gamma_1 \gamma_2 |a_c|}{1-|a_c|^2}.\label{power}
    \end{align}
\end{lemma}
\begin{proof}
    Calculating the total power of the derived minimum norm solution from Lemma \ref{lem1}, we obtain 
    \begin{align}
        \bw^H\bw &=  \left(\bA(\bA^H\bA)^{-1} \boldg\right)^H\bA(\bA^H\bA)^{-1} \boldg\\
        &= \boldg^H(\bA^H\bA)^{-1} \boldg.
    \end{align}
    The inverse of the Hermitian Gram matrix $\bA^H\bA$ can be written as
    \begin{align}
        \left(\bA^{H}\bA\right)^{-1} &= 
        \begin{bmatrix}
            1&a_{c}\\
            a^*_{c}&1
        \end{bmatrix}^{-1} \\&= 
        \frac{1}{1-|a_{c}|^2}
        \begin{bmatrix}
            1&-a_{c}\\
            -a^*_{c}&1
        \end{bmatrix},\label{Gram}
    \end{align}
    where $a_c$ is the cross correlation of the steering vectors. Inserting \eqref{Gram} into the expression for $\bw^H\bw$ above yields
    \begin{align}
        \bw^H\bw &= \frac{\gamma_1^2+\gamma_2^2 -2\gamma_1\gamma_2 \Re\left[e^{j\Delta\phi}a_{c}\right]}{1-|a_{c}|^2}, \label{power2}
    \end{align}
    where $\Delta \phi = \phi_1-\phi_2$ and $\Re\left[\cdot\right]$ is the real part operator.
    The optimization variables $\phi_1$ and $\phi_2$ can be chosen such that $e^{j\Delta\phi}a_{c}$ is real valued which consequently minimizes the expression in \eqref{power2}, i.e., $-\angle a_c$. Hence, the minimum total transmit power is given by \eqref{power}. This completes the proof.
\end{proof}
Lemmas \ref{lem1} and \ref{lem2} provide the optimal precoder that minimizes the total transmit power while controlling the power allocation in both directions through the parameters $\gamma_1$ and $\gamma_2$. Alternatively, if the total transmit power is fixed to $P_\text{total}$, we can parameterize the power split as $\gamma_1 = s\gamma$ and $\gamma_2 = s(1-\gamma)$, where $\gamma \in (0,1)$ and the scaling factor $s$ is given by
\begin{align}
    s= \sqrt{\frac{P_\text{total}\left(1-|a_c|^2\right)}{1-2\gamma\left(1+|a_c|\right)+2\gamma^2\left(1+|a_c|\right)}}.
\end{align}
With this parameterization, the resulting precoder for an arbitrary $\gamma$ becomes
\begin{align}
    \bw = s\bA(\bA^H\bA)^{-1} \left[\gamma e^{-j\arg(a_c)},(1-\gamma)\right]^T.\label{beamformer}
\end{align}

\subsection{Sensing Strategies}
Utilizing the derived two path precoder, we propose three strategies on how to divide the available sensing power between the direct and the RIS-assisted sensing path. Hence, in the sequel $\theta_1 = \theta_\text{D}$ represents the direction directly towards the target, whereas $\theta_2 = \theta_\text{RIS}$ represents the direction towards the RIS, resulting in $\bA = [\ba_\text{D},\ba_\text{RIS}]$ where the columns are the steering vectors directed towards $\theta_\text{D}$ and $\theta_\text{RIS}$, respectively.
\paragraph{Temporal Separation} In this approach, the radar alternates between beamforming towards the direct link and the RIS-assisted link over the slow-time dimension. The temporal separation transmit beamformer is therefore given by
\begin{align}
    \bw[l] = 
    \begin{cases}
        \bw_\text{D} = \bA\bigl(\bA^H\bA\bigr)^{-1} 
        \begin{bmatrix}
            s&0
        \end{bmatrix}^T
        & \text{if} \; l\in \mathcal{T}_\text{D}, \\ 
        \bw_\text{RIS} = \bA\bigl(\bA^H\bA\bigr)^{-1} 
        \begin{bmatrix}
            0&s
        \end{bmatrix}^T & \text{if} \; l\in \mathcal{T}_\text{RIS},
    \end{cases}
    \label{tempsep}
\end{align}
where $\mathcal{T}_\text{D}$ and $\mathcal{T}_\text{RIS}$ represent the sets of slow-time samples within the time sections allocated towards the direct and RIS paths, respectively.

\paragraph{Spatial Separation}
In this algorithm, the radar simultaneously transmits towards the direct path and the RIS-generated path. Therefore, the composite beamformer can be constructed as given in \eqref{beamformer}. The parameter $\gamma$ determines the power portion allocated towards the direct path.

\paragraph{Opportunistic Selection}
In this case, we choose only one of the direct or RIS-generated paths by setting $\gamma = 0$ and $\gamma = 1$, respectively. Based on the quality of vital sign extracted previously, i.e., we stay with the currently selected path until the quality of the extracted vital sign estimates drops, then we switch to the other path. 

Each of the three proposed strategies has advantages and disadvantages. The temporal separation approach achieves the highest possible SNR for the active path by focusing all transmit power on it. This benefit is offset by a reduced observation window for each path, which can degrade the frequency resolution of the breathing signal or limit the maximum detectable frequency as per the Nyquist criterion. In contrast, spatial separation ensures continuous information from both links by transmitting to them simultaneously, thereby avoiding the time-frequency limitations of the temporal method. However, this requires splitting the total power, which can lower the SNR for each path. Finally, opportunistic selection offers an adaptive solution that maximizes measurement quality at any given time. The disadvantage of this method is the formulation of the criterion to initiate the switch, the possibility of wrong path selection, and the overhead of path selection. 

\subsection{Vital Sign Extraction}
\label{VitalSignExtraction}
In what follows, we introduce the processing chain that converts the recorded radar data into two one-dimensional respiration signals, representing the observation via the RIS path and the direct path.
\subsubsection{Data acquisition}
The receive signal $\Tilde{\by}_{m} \in \mathbb{C}^{K}$ at each of the $M$ antennas is matched filtered with the transmit signal for each slow-time sample $l$ as
\begin{align}
    \Tilde{y}_{m}(l) = \Tilde{\by}_{m}(l)\, \bs^H.
\end{align}
The resulting signal $\Tilde{y}_{m}(l)$ is then arranged in the data matrix $\Tilde{\bY} \in \mathbb{C}^{M \times L}$.
\subsubsection{Clutter suppression}
Static reflections are removed by a slow-time high‑pass filter of width $W$ as
\begin{align}
    \bY[m,l]=\Tilde{\bY}[m,l]-\frac{1}{W}\sum_{w=-\frac{W-1}{2}}^{\frac{W-1}{2}} \Tilde{\bY}[m,l+w],
\end{align}
resulting in the filtered signal matrix $\bY \in \mathbb{C}^{M \times L}$.

\subsubsection{Path separation by beamforming}\label{Step4} 
The decoded receive signal vector $\br \in \mathbb{C}^L$ is given by
\begin{align}
    \br^{(\text{D})}= \mathbf w_\mathrm{D}^{\mathrm H}\,\bY,
\qquad
\br^{(\text{RIS})}= \mathbf w_\mathrm{RIS}^{\mathrm H}\,\bY,
\end{align}
where $\br = [r_1, \hdots r_l \hdots r_L]^T$ and the receive beamfromers $\bw_\text{D}$ and $\bw_\text{RIS}$ are given in \eqref{tempsep}.

\subsubsection{Phase demodulation} 
Unwrapping the argument gives the slow-time phase sequences
\begin{align}
    \varphi_\mathrm{D}[l]=\angle\bigl(\br^{(\text{D})}[l]\bigr),\qquad
    \varphi_\mathrm{RIS}[l]=\angle\bigl(\br^{(\text{RIS})}[l]\bigr).
\end{align}
After removing linear trends caused by unwrapping, the chest displacement is obtained as
\begin{align}
    d_\mathrm{D}[l]=\frac{1}{2}\cdot\frac{\lambda_0}{2\pi}\,\varphi_\mathrm{D}[l],\qquad
    d_\mathrm{RIS}[l]=\frac{1}{2}\cdot\frac{\lambda_0}{2\pi}\,\varphi_\mathrm{RIS}[l].
\end{align}
The factor of $\frac{1}{2}$ is due to the fact that the signal travels distance $d_\text{D}$ or $d_\text{RIS}$ twice. 

\subsubsection{Evaluation and reconfiguration}
Each chest displacement sequence is transformed into the frequency domain, and its power spectrum is searched for a dominant respiratory component. The sensing branch whose spectrum provides the tallest and most isolated respiration peak is granted a larger share of sensing resources in the next data acquisition period. If neither sensing path provides a clear respiration peak, the patient's position estimate is repeated and the RIS phase shifts, and steering vectors updated accordingly.

Fig. \ref{fig:Method} depicts an overview of the proposed algorithm. First, the position of the person is estimated using state-of-the-art algorithms, here root-MUSIC. The estimated position is used to configure the RIS-phase shifts and the transmit and receive beamformers. For the recording of the first batch of $L$ pulses, the radar sends the sensing waveform initially precoded with equal weight on both paths towards the target. The echo signal is processed according to Step 1) to 4). Finally, the vital sign signals from both paths are analyzed in the frequency domain and the path weighting is updated as described in Step 5).   

\begin{figure}
    \centering
    \resizebox{\linewidth}{!}{
        \begin{tikzpicture}[
    block/.style = {rectangle, draw, minimum width=2.0cm, minimum height=1.0cm, align=center},
    picblock/.style = {rectangle, draw, minimum width=2.2cm, minimum height=1.3cm, black, fill= blue!20!white},
    arrow/.style = {->, thick}
]

\node[block, align=center] (A1) {Position\\Estimation};
\node[block, align=center, right=0.8cm of A1] (A2) {RIS \& Precoder\\Initialization};
\node[block, align=center, right=0.8cm of A2] (B) {Resource\\Allocation};
\node[block, align=center, right=1.2cm of B] (C) {Radar\\Measurement};
\node[block, align=center, below=3.0cm of C] (D) {
Signal\\Evaluation};

\draw[arrow] (A1) -- (A2);
\draw[arrow] (A2) -- (B);
\draw[arrow] (B) -- (C);

\node[picblock, below left=0.8cm and 0.2cm of C.south] (F1) {};
\node[picblock, below right=0.8cm and 0.2cm of C.south] (F2) {};

\draw[arrow] ([shift={(-0.5,0)}]A1.west)  -- (A1.west) node[pos=0, left, align=center]{Start};
\draw[arrow] (C.south)++(-0.65,0) -| ([shift={(0.65,0)}]F1.north);
\draw[arrow] (C.south)++(0.65,0) -| ([shift={(-0.65,0)}]F2.north);
\draw[arrow] (F1.south)++(0.67,0) -| ([shift={(-0.65,0)}]D.north);
\draw[arrow] (F2.south)++(-0.67,0) -| ([shift={(0.65,0)}]D.north);
\draw[arrow] (D.west) -| (B.south) 
    node[pos=0.75, above, align=center, rotate=90]{Precoding\\Correction};
\draw[arrow] (D.west) -| (A1.south) 
    node[pos=0.74, above, align=center, rotate=90]{Position\\Correction};
\draw[arrow] (D.east) -- ([shift={(0.5,0)}]D.east) node[pos=1, right, align=center]{Vital Sign\\Frequency\\Estimate};

\draw[arrow] (F1.south west) ++(0.2,0.2) -- ++(1.55,0) node[right] {$f$};
\draw[arrow] (F1.south west) ++(0.2,0.2) -- ++(0,1.0) node[pos=0.5, align=left, right] {VS\,Spectrum\\RIS Path};

\draw[arrow] (F2.south west) ++(0.2,0.2) -- ++(1.55,0) node[right] {$f$};
\draw[arrow] (F2.south west) ++(0.2,0.2) -- ++(0,1.0) node[pos=0.5, align=left, right] {VS\,Spectrum\\Direct Path};

\end{tikzpicture}
    }
    \caption{Overview of the proposed vital sign extraction algorithm.}
    \label{fig:Method}
\end{figure}
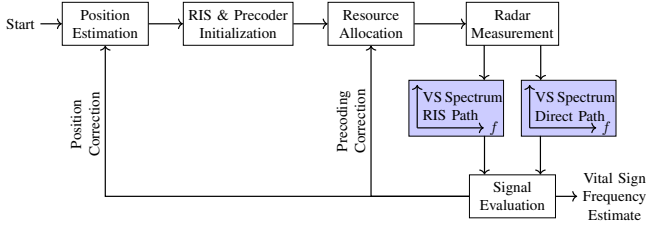

\section{Performance Evaluation and Experiment}
This section quantifies the benefits of exploiting two spatially diverse
sensing paths using the proposed algorithm strategies. First, we detail the hybrid simulation-experimental model used to emulate a small indoor patient scenario. Then, we discuss the obtained results to assess the performance of the proposed schemes.
 
\subsection{Simulation Setup}
A monostatic phased array radar with $M=5$ collocated transceivers is adopted. The radar operates at $f_\mathrm{c}=7.15\,\text{GHz}$ with a bandwidth of $B=0.5\,\text{MHz}$ sampled with $K=64$ fast-time samples. The sensing pulses are
repeated every $T_\mathrm{PRI}=250\,\text{ms}$, resulting in a
\mbox{$4\,\text{Hz}$} slow-time sampling rate. The total transmit power is $P_\text{total}=10\,\text{mW}$, and the receiver has a $10\,\text{dB}$ noise figure, which produces a thermal noise floor of $\sigma^2_\text{n}=-107\,\text{dBm}$.
\begin{figure}[t]
    \centering
    \resizebox{\linewidth}{!}{
        \begin{tikzpicture}
    \begin{axis}[%
      view       = {35}{20},      
      xlabel     = {Width in m},
      ylabel     = {Depth in m},
      zlabel     = {Height in m},
      enlargelimits = 0.05,
      grid=both,
      xmin=0, xmax=3.5,
      ymin=-1, ymax=1.4606,
      zmin=0, zmax=2,
      axis lines = left,
      ticklabel style = {font=\small},
      legend style = {font=\small, at={(1.05,0.8)}, anchor=north west},
      axis equal image,
    ]

      \addplot3[
        only marks,
        mark=*,
        mark size=2.5pt,
        black,
      ] coordinates {(0,0,1)}
        node[anchor=south] {Radar};
        \addlegendentry{Radar}

      \addplot3[
        only marks,
        mark=triangle*,
        mark options={scale=3,solid,orange}
      ] coordinates {(3,0,1)} node[anchor=north] {Patient};
      \addlegendentry{Patient}

      \def\ctr{(2.707, 1.4606, 1)}   
      \def\h{0.21}             
      \addplot3[
        very thick,
        green!60!black,
        forget plot
      ] coordinates {
        (2.707-\h, 1.4606, 1-\h)
        (2.707+\h, 1.4606, 1-\h)
        (2.707+\h, 1.4606, 1+\h)
        (2.707-\h, 1.4606, 1+\h)
        (2.707-\h, 1.4606, 1-\h)
      };
      \addlegendimage{only marks, mark=square*, mark options={draw=green!60!black, fill=white, thick}}
      \addlegendentry{RIS}
      \addplot3[
        only marks,
        mark=none,
        mark size=0pt,
        green!60!black,
        forget plot,
      ] coordinates {(2.707,1.4606,1.2)} node[anchor=south, rotate=-15] {RIS};

      \addplot3[
        blue,
        dashed,
        very thick,
      ] coordinates {(0,0,1) (3,0,1)};
      \addlegendentry{Direct Path}

      \addplot3[
        red,
        dashed,
        very thick,
        forget plot,
      ] coordinates {(0,0,1) (2.707,1.4606,1)};

      \addplot3[
        red,
        dashed,
        very thick,
      ] coordinates {(2.707,1.4606,1) (3,0,1)};
      \addlegendentry{RIS Path}

    \end{axis}
  \end{tikzpicture}
    }
    \caption{Considered simulation setup in which the patient's chest is facing the RIS.}
    \label{fig:SimulationSetup}
\end{figure}
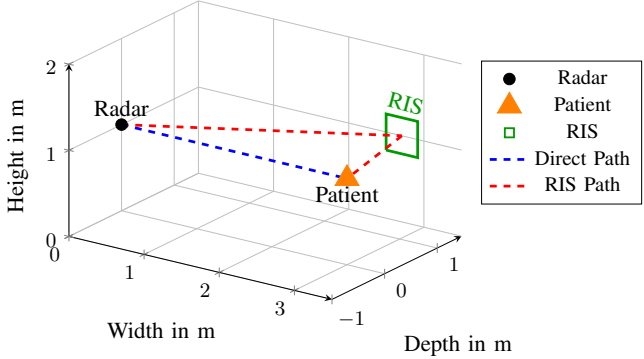
Fig.~\ref{fig:SimulationSetup} sketches the spatial
layout. The patient's chest is placed $3\,\text{m}$ in front of the radar and faces an RIS panel that is mounted $1.5\,\text{m}$ in front of the chest. The direct path therefore reaches the chest at an incidence angle of $78.75^{\circ}$ relative to its normal, whereas the RIS illuminates the chest from the front.
The target RCS is obtained from experimental measurements and used in the simulated RIS-assisted radar environment. 
The RIS consists of an $10\times10$ antenna grid with half wavelength spacing leading to a 21\,cm square panel. Each element applies an ideal frequency phase shift such that the incident wave from the radar is redirected and focused at the patient's chest. 
The large scale channel experiences Rician fading with a
$K$ factor of $10\,\text{dB}$, representative of a rich multipath indoor environment but still preserving a dominant direct component.

\subsection{Experimental Acquisition of Respiration Data}
To emulate realistic chest micro motions in the simulation, we collected an
experimental respiration data set.  Two stepped frequency continuous wave
(SFCW) radars (Walabot Developer, Vayyar) recorded simultaneously the breathing cycle of a proband as shown in Fig. \ref{fig:ExperimentPicture}.  
One radar was located orthogonally to the subject's chest at a
range of 2\,m and serves as the front view reference.
The second radar was successively repositioned on a circular arc of equal
radius, covering aspect angles from 11.25\(^{\circ}\) (close to the first radar) to 90\(^{\circ}\) in 8\(^{\circ}\) steps.  Each run lasted
60\,s so that both sensors captured identical respiration
cycles.

The raw radar signals were processed as described in Section \ref{VitalSignExtraction}.
The resulting angle dependent displacement traces constitute the ground truth respiration waveforms that are injected into the simulation.

\begin{figure}
    \centering
    \includegraphics[width=0.5\linewidth]{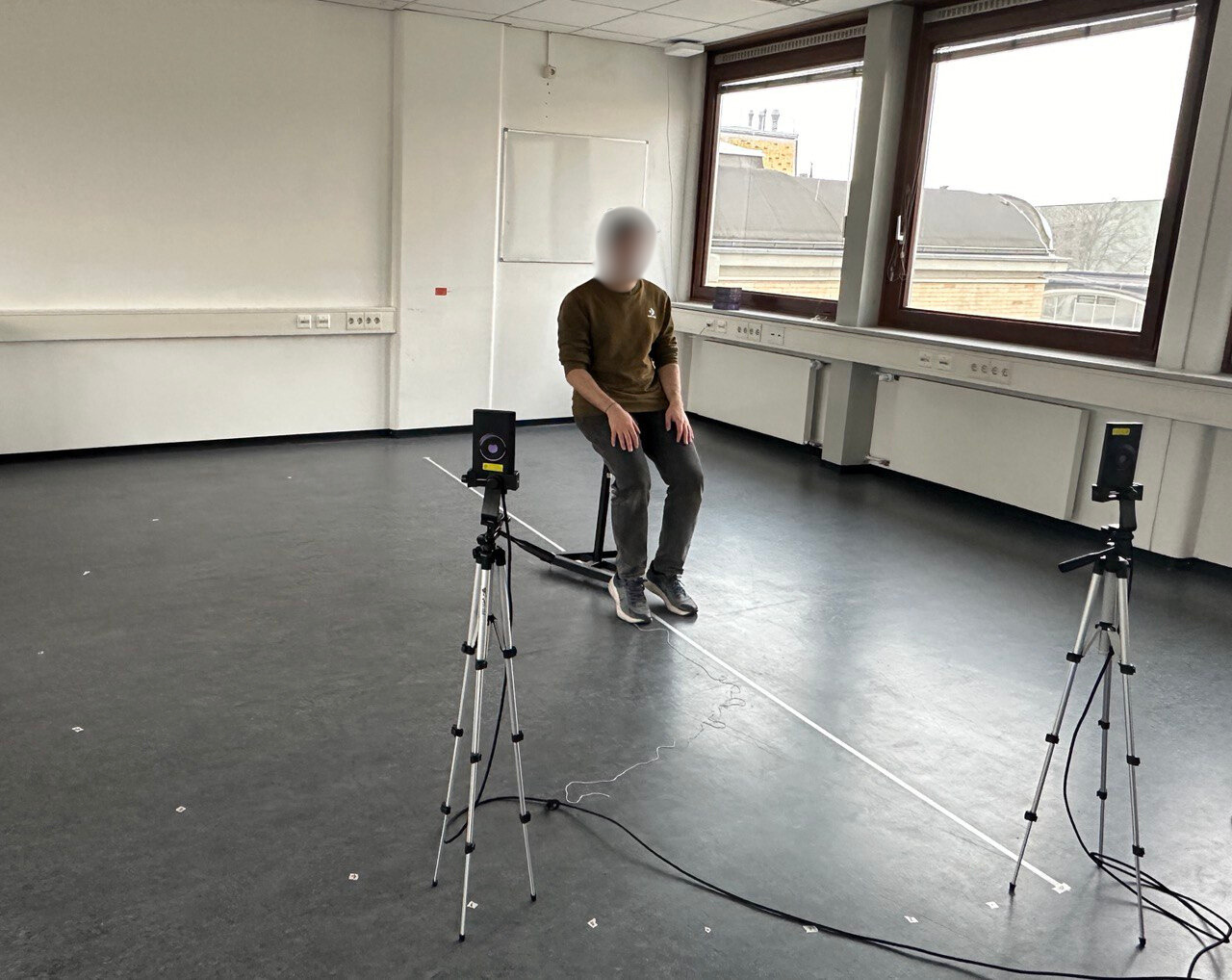}
    \caption{Setup for the experimental radar data acquisition involving two radars.}
    \label{fig:ExperimentPicture}
\end{figure}

\subsection{Simulation Results}

We now quantify how the three proposed sensing strategies exploit the
two spatially diverse paths provided by the RIS.  Unless stated
otherwise, the resource allocation parameter is set to $\gamma=0.5$ so that the
direct and RIS paths share the sensing resources equally.

\label{sec:sim_results}
\subsubsection{Time Domain Performance}
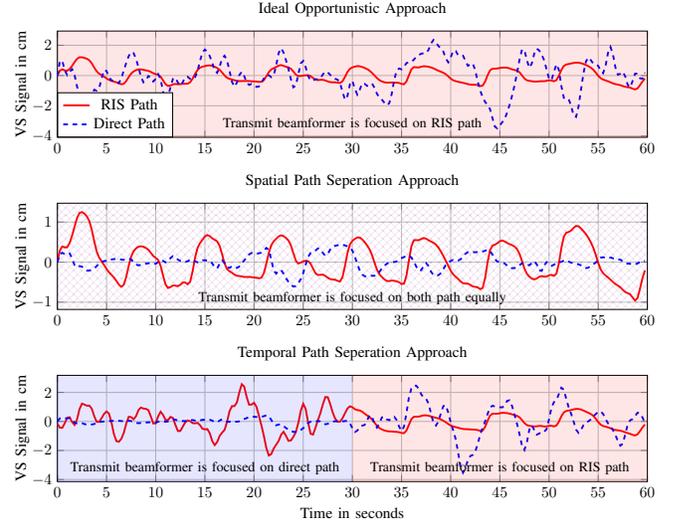
\begin{figure}
    \centering
    \resizebox{\linewidth}{!}{
        \begin{tikzpicture}
\begin{groupplot}[
    group style={
        group size=1 by 3,       
        vertical sep=1.5cm,        
    },
    scaled y ticks=false,
    yticklabel style={
      /pgf/number format/fixed,
      /pgf/number format/precision=2
    },
    yticklabel={\pgfmathparse{\tick*100}\pgfmathprintnumber{\pgfmathresult}},
    y label style={yshift=-0.5em},
    ylabel={VS Signal in cm},
    width=15cm,
    height=4cm,
    xmin=0, xmax=60,
    ylabel={VS Signal in cm},
    grid=both,
    legend style={at={(-0.0,-0.0)},anchor=south west},
]

\nextgroupplot[
    title={Ideal Opportunistic Approach}
]
\addplot[red, very thick]
    table [
        x=time_s,
        y=Gini_IRS_m,
        col sep=comma
    ] {figures/data/respiration_time_domain.csv};

\addplot[blue, dashed, very thick]
    table [
        x=time_s,
        y=Gini_Direct_m,
        col sep=comma
    ] {figures/data/respiration_time_domain.csv};
    
    \fill[red, opacity=0.1] (axis cs:0,-0.6) rectangle (axis cs:60,0.6);

    \node[anchor=south] at (axis cs:30,-0.04) {\small Transmit beamformer is focused on RIS path};

    \legend{RIS Path, Direct Path}
\nextgroupplot[
    title={Spatial Path Seperation Approach}, 
]
\addplot[red, very thick]
    table [
        x=time_s,
        y=PowerSplit_IRS_m,
        col sep=comma
    ] {figures/data/respiration_time_domain.csv};

\addplot[blue, dashed, very thick]
    table [
        x=time_s,
        y=PowerSplit_Direct_m,
        col sep=comma
    ] {figures/data/respiration_time_domain.csv};

    \fill[
        pattern=north east lines,
        pattern color=blue,
        opacity=0.3
    ] (axis cs:0,-0.6) rectangle (axis cs:60,0.6);

    \fill[
        pattern=north west lines,
        pattern color=red,
        line width = 10mm,
        opacity=0.3
    ] (axis cs:0,-0.6) rectangle (axis cs:60,0.6);

    \node[anchor=south] at (axis cs:30,-0.012) {\small Transmit beamformer is focused on both path equally};

\nextgroupplot[
    title={Temporal Path Seperation Approach}, 
    xlabel={Time in seconds},
]
\addplot[red, very thick]
    table [
        x=time_s,
        y=TimeSplit_IRS_m,
        col sep=comma
    ] {figures/data/respiration_time_domain.csv};

\addplot[blue, dashed, very thick]
    table [
        x=time_s,
        y=TimeSplit_Direct_m,
        col sep=comma
    ] {figures/data/respiration_time_domain.csv};

    \fill[blue, opacity=0.1] (axis cs:0,-0.6) rectangle (axis cs:30,0.6);
    \fill[red, opacity=0.1] (axis cs:30,-0.6) rectangle (axis cs:60,0.6);

    \node[anchor=south] at (axis cs:15,-0.04) {\small Transmit beamformer is focused on direct path};
    \node[anchor=south] at (axis cs:45,-0.04) {\small Transmit beamformer is focused on RIS path};

\end{groupplot}
\end{tikzpicture}
    }
    \caption{Vital sign signals obtained from the direct and RIS-assisted path for the three proposed sensing strategies. The resources are shared equally for the spatial and temporal separation, whereas the ideal opportunistic approach has prior knowledge about the optimal path.}
    \label{fig:ResultTime}
\end{figure}

Fig. \ref{fig:ResultTime} depicts the slow-time vital-sign signals
recovered from the direct (blue) and RIS (red) paths.  
For the ideal opportunistic strategy (top panel), the transmit
beam is steered entirely toward the RIS, resulting in a clean
sinusoidal displacement of approx. $2\,\mathrm{cm}$ peak-to-peak whose
period matches the breathing cycle.
In contrast, the direct path trace exhibits
noise like fluctuations.

The spatial separation strategy (middle panel) illuminates both links at the same time but with only half the SNR.
The RIS branch still exhibits a clear, sinusoidal breathing trace, while the direct branch no longer looks like random noise. Instead, it shows a weak, irregular displacement that mainly stems from shoulder motion rather than chest expansion.
With temporal separation (bottom panel) the radar alternates the full transmit power between the two beams.
Each path therefore cycles through high-quality segments and low-quality segments, resulting in the alternating appearance seen in the plot.

\begin{figure}
    \centering
    \resizebox{\linewidth}{!}{
        \begin{tikzpicture}
\begin{groupplot}[
    group style={
        group size=3 by 1,       
        vertical sep=2cm,        
    },
    width=5cm,
    height=5cm,
    xmin=0, xmax=0.3,
    grid=both,
    legend style={at={(-0.07,-0.2)},anchor=north},
    xlabel={Frequency in Hz},
]

\nextgroupplot[
    title={Ideal Opportunistic},
    ylabel={Spectral Estimate},
]
\addplot[red, very thick]
    table [
        x=freq_Hz,
        y=Gini_IRS,
        col sep=comma
    ] {figures/data/respiration_frequency_domain.csv};

\addplot[blue, dashed, very thick]
    table [
        x=freq_Hz,
        y=Gini_Direct,
        col sep=comma
    ] {figures/data/respiration_frequency_domain.csv};

\nextgroupplot[
    title={Spatial Seperation}, 
]
\addplot[red, very thick]
    table [
        x=freq_Hz,
        y=PowerSplit_IRS,
        col sep=comma
    ] {figures/data/respiration_frequency_domain.csv};

\addplot[blue, dashed, very thick]
    table [
        x=freq_Hz,
        y=PowerSplit_Direct,
        col sep=comma
    ] {figures/data/respiration_frequency_domain.csv};

\nextgroupplot[
    title={Temporal Seperation}, 
]
\addplot[red, very thick]
    table [
        x=freq_Hz,
        y=TimeSplit_IRS,
        col sep=comma
    ] {figures/data/respiration_frequency_domain.csv};

\addplot[blue, dashed, very thick]
    table [
        x=freq_Hz,
        y=TimeSplit_Direct,
        col sep=comma
    ] {figures/data/respiration_frequency_domain.csv};

\end{groupplot}
\end{tikzpicture}
    }
    \caption{Frequency spectra of the obtained direct (blue) and RIS-assisted (red) vital sign signal  for each sensing strategy.}
    \label{fig:ResultFreq}
\end{figure}
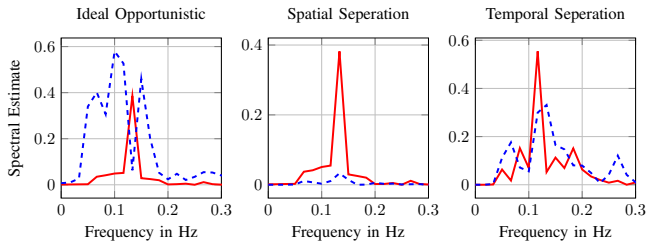
\subsubsection{Spectral accuracy}
The corresponding power spectra in Fig. \ref{fig:ResultFreq} confirm
these observations.  The opportunistic scheme
has a pronounced peak at $0.133\,\text{Hz}$ 
with high prominence over the noise floor, while the direct-path spectrum is almost flat.
Spatial separation produces a $0.133\,\text{Hz}$ peak on the RIS branch at $13\,$dB, whereas the direct branch shows a weaker and less prominent peak at the same location.
Temporal separation reveals clear peaks on both paths but with reduced
frequency resolution (wider main lobes) due to the shortened observation window.
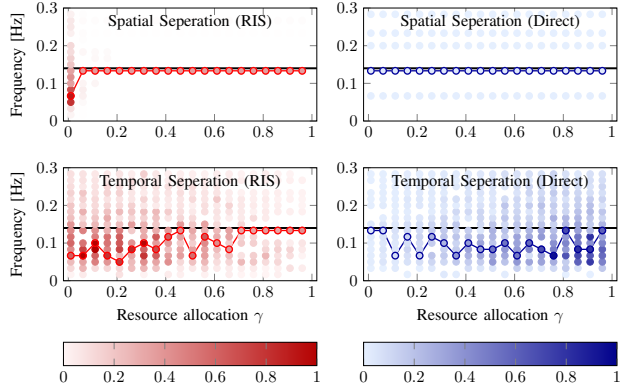
\begin{figure}
    \centering
        \hspace*{-2.2cm}
    \resizebox{1.2\linewidth}{!}{
        \usepgfplotslibrary{groupplots}
\centering
\begin{tikzpicture}
\begin{groupplot}[
    group style={
        group size=2 by 2,
        horizontal sep=1.0cm,
        vertical sep=1.0cm
    },
    ymin=0, ymax=0.3,
    xmin=-0.02, xmax=1.02,
    width=7cm,
    height=4cm,
]

\nextgroupplot[title={Spatial Seperation (RIS)}, 
title style={at={(0.5,0.9)},anchor=north}, ylabel={Frequency [Hz]}, colormap={custom}{rgb255(0cm)=(255,255,255) rgb255(1cm)=(180,0,0)}, point meta=explicit, point meta min=0, point meta max=1]
\addplot[
    only marks,
    mark=*,
    scatter,
    scatter src=explicit,
    scatter/use mapped color={draw=mapped color, fill=mapped color},
    on layer=axis background,
] table [x=Split, y=Frequency_Hz, meta=Power, col sep=comma] {figures/data/Power_IRS.csv};

\addplot[
    color=red,
    mark=o, 
    thick
] table [
    x=Split,
    y=Frequency_Hz,
    col sep=comma
] {figures/data/Maxima_Power_IRS.csv};
\addplot[
    black,
    very thick,
    domain=-0.02:1.02,
    samples=2
] {0.14};

\nextgroupplot[title={Spatial Seperation (Direct)}, title style={at={(0.5,0.9)},anchor=north}, colormap={custom}{rgb255(0cm)=(230,240,255) rgb255(1cm)=(0,0,139)}, point meta=explicit, point meta min=0, point meta max=1]
\addplot[
    only marks,
    mark=*,
    scatter,
    scatter src=explicit,
    scatter/use mapped color={draw=mapped color, fill=mapped color},
    on layer=axis background,
] table [x=Split, y=Frequency_Hz, meta=Power, col sep=comma] {figures/data/Power_direct.csv};

\addplot[
    color=blue!60!black,
    mark=o, 
    thick
] table [
    x=Split,
    y=Frequency_Hz,
    col sep=comma
] {figures/data/Maxima_Power_direct.csv};
\addplot[
    black,
    very thick,
    domain=-0.02:1.02,
    samples=2
] {0.14};

\nextgroupplot[
    title={Temporal Seperation (RIS)}, title style={at={(0.5,0.9)},anchor=north}, 
    ylabel={Frequency [Hz]},
    xlabel={Resource allocation $\gamma$},
    colormap={custom}{rgb255(0cm)=(255,245,245) rgb255(1cm)=(180,0,0)},
    colorbar horizontal,
    point meta=explicit, point meta min=0, point meta max=1
]
\addplot[
    only marks,
    mark=*,
    scatter,
    scatter src=explicit,
    scatter/use mapped color={draw=mapped color, fill=mapped color},
    on layer=axis background,
] table [x=Split, y=Frequency_Hz, meta=Power, col sep=comma] {figures/data/Time_IRS.csv};

\addplot[
    color=red,
    mark=o, 
    thick
] table [
    x=Split,
    y=Frequency_Hz,
    col sep=comma
] {figures/data/Maxima_Time_IRS.csv};
\addplot[
    black,
    very thick,
    domain=-0.02:1.02,
    samples=2
] {0.14};

\nextgroupplot[
    title={Temporal Seperation (Direct)}, title style={at={(0.5,0.9)},anchor=north}, 
    colormap={custom}{rgb255(0cm)=(230,240,255) rgb255(1cm)=(0,0,139)},
    xlabel={Resource allocation $\gamma$},
    colorbar horizontal,
    point meta=explicit, point meta min=0, point meta max=1
]
\addplot[
    only marks,
    mark=*,
    scatter,
    scatter src=explicit,
    scatter/use mapped color={draw=mapped color, fill=mapped color},
    on layer=axis background,
] table [x=Split, y=Frequency_Hz, meta=Power, col sep=comma] {figures/data/Time_direct.csv};

\addplot[
    color=blue!60!black,
    mark=o, 
    thick
] table [
    x=Split,
    y=Frequency_Hz,
    col sep=comma
] {figures/data/Maxima_Time_direct.csv};
\addplot[
    black,
    very thick,
    domain=-0.02:1.02,
    samples=2
] {0.14};

\end{groupplot}
\end{tikzpicture}
    }
    \caption{The colored markers represent peaks of the vital sign spectra for different values of $\gamma$ using the spatial and temporal separation strategy. The peaks are color coded with respect to their prominence. The black line indicates the true breathing frequency, whereas the red and blue lines indicate the progression of the maximum peak with respect to $\gamma$ for the RIS-assisted and direct path, respectively.}
    \label{fig:ResultFreqComp}
\end{figure}
\subsubsection{Impact of resource allocation}
Fig. \ref{fig:ResultFreqComp} shows that the dominant spectral peak locks onto the true breathing rate as soon as the RIS branch receives roughly four-fifths of the power ($\gamma \geq 0.2$) in spatial splitting, while the direct-path beam then exhibits only a faint replica caused by multipath leakage. In temporal splitting the RIS peak stabilizes only after it is allotted more than half of the slow-time slots ($\gamma \geq 0.7$), and slots aimed at the direct path never form a clear peak. Hence even a modest bias toward the RIS path yields a consistent respiration estimate, underscoring the benefit of RIS-enabled diversity.
\section{Conclusion}
We introduced the first radar architecture that leverages an RIS to create two controllable
sensing paths for contactless vital sign monitoring.
A hybrid analytic/experimental channel/RCS framework enabled investigating a closed-form dual path beamformer and three resource
allocation strategies.  
Simulations driven by real respiration traces show that the exploitation of the spatial diversity provided by the RIS enables more resilient vital sign monitoring for cases in which the person's chest does not directly face the radar. 
Because the RIS requires no RF chains and no fronthaul, the proposed
scheme offers a cost-effective alternative to multi-radar networks for in-home and clinical monitoring. 
\bibliographystyle{IEEEtran}
\bibliography{References}
\end{document}